\documentclass{llncs}
\usepackage{llncsdoc}
\usepackage[english]{babel}
\usepackage{amsmath}
\usepackage{amssymb,amsfonts,textcomp}
\usepackage{multicol}
\usepackage{array}
\usepackage{hhline}
\begin{document}

\title{Analysis Of The Girth For Regular Bi-partite Graphs With Degree $3$}

\author{Vivek S Nittoor \and Reiji Suda}

\institute{The University Of Tokyo}

\maketitle
\begin{abstract}
The goal of this paper is to derive the detailed description of the Enumeration Based Search Algorithm from the high level description provided in $[16]$, analyze the experimental results from our implementation of the Enumeration Based Search Algorithm for finding a regular bi-partite graph of degree $3$, and compare it with known results from the available literature. We show that the values of $m$ for a given girth $g$ for $(m, 3)$ BTUs are within the known mathematical bounds for regular bi-partitite graphs from the available literature. 
\end{abstract}

\section{Introduction}
The goal of this paper is to develop the detailed description of the Enumeration Based Search Algorithm from the high level description provided in $[16]$ and analyze the implementation results of the Enumeration Based Search Algorithm for finding a regular bi-partite graph of degree $3$, and compare it with known results from the available literature. $(m, r)$ BTU is our notation for a regular bi-partite graph that has been introduced in $[1]$. The high level description of the Enumeration Based Search Algorithm for searching a girth maximum $(m, r)$ BTU has been described in $[16]$. The theoretical background behind BTUs has been introduced and  explained in detail in $[1]$ and $[2]$.

\section {Girth Maximization as a Extremal Graph Theory question}

We consider the problem of searching for a girth maximum $(m,r)$ BTU as a question in Extremal Graph Theory by raising two related questions.
\begin{enumerate}
\item Given girth $g$ and  $r\in \mathbb{N}$ , what is the minimum value of  $m$ such that a $(m,r)$ BTU has girth $g$  .
\item Given  $m,r\in \mathbb{N};m\gg r$, what is the maximum attainable girth for a $(m,r)$ BTU ?
\end{enumerate}

\subsection{Definitions}
We review definitions from  $[1]$ and  $[2]$ .

\begin{definition} $(m, r)$ BTU \\
A $(m,r)$ Balanced Tanner Unit (BTU) is a regular bi-partite graph that can be represented by a $m \times m$ square matrix with $r$ non-zero elements in each of its rows and columns. Every $(m,r)$ BTU has a bipartite graph representation and an equivalent matrix representation.
\end{definition}

\begin{definition} {Girth maximum $(m,r)$ BTU} \\
{A labeled $(m,r)$  BTU  $A$ is girth maximum if there does not exist another labeled $(m,r)$  BTU $B$ with girth greater than that of $A$ .}
\end{definition}

\begin{definition} $\Phi (\beta _{1},\beta _{2},\ldots ,\beta _{r - 1})$ where $\beta _{i}\in P_{2}(m)$ for  $1\le i\le r-1$  \\
 $\Phi (\beta _{1},\beta _{2},\ldots ,\beta _{r\text{--}1})$ refers to the family of all labeled $(m,r)$ BTUs with compatible permutations  $p_{1,}p_{2,}\ldots ,p_{r}\in S_{m};p_{i}\notin C(p_{1},p_{2},\ldots ,p_{i-1})$ for  $1<i\le r$ that occur in the same order on a complete $m$ symmetric permutation tree ,  $x_{1,1}<x_{2,1}<\ldots <x_{r,1}$ where  $p_{j}=(x_{j,1}x_{j,2}\ldots x_{j,m});1\le j\le r$ , such that  $\beta _{i-1}$  is the partition between permutations $p_{i\text{--}1}$  and  $p_{i}$  for all integer values of  $i$  given by  $1<i\le r$ .
\end{definition}
\begin{definition} {Optimal partition parameters for girth maximum $(m, r)$ BTU}.
$\beta _{1},\beta _{2},\ldots ,\beta _{r-1}\in P_{2}(b\ast k^{r-1})$ refer to optimal partitions derived in $[2]$ such that there exists a girth maximum $(m,r)$ BTU \ in $\Phi (\beta _{1},\beta _{2},\ldots ,\beta _{r-1})$ ,where $\beta _{i}$ refers to $\sum _{j=1}^{r-1-i}b\ast k^{i}=b\ast k^{r-1}$ for $1\le i\le r-1$, with $k \in \mathbb{N}$ obtained as a solution to $b\ast k^{r -1} = m$ such that $b \in \mathbb{N}$ is minimized. Thus,  $\beta _{1},\beta _{2},\ldots ,\beta _{r-1}\in P_{2}(b\ast k^{r-1})$  are $\sum _{j=1}^{k^{r-2}}b\ast k=b\ast k^{r-1}$ , $\sum _{j=1}^{k^{r-3}}b\ast k^{2}=b\ast k^{r-1}$ , $\ldots $, $\sum _{j=1}^{k}b\ast k^{r-2}=b\ast k^{r-1}$ and $\sum _{j=1}^{1}b\ast k^{r-1}=b\ast k^{r-1}$ respectively.
\end{definition}

\subsection{Search for girth maximum $(m, r)$ BTU }
Search for a girth maximum $(m, r)$ BTU refers to search for an optimal labelled  $(m,r)$ BTU in a family of labelled BTUs that we refer to as  
$\Phi (\beta _{1},\beta _{2},\ldots ,\beta _{r - 1})$ where $\beta _{i}\in P_{2}(m)$ for  $1\le i\le r-1$. 

\section {Girth Maximization as a Extremal Graph Theory question}

We consider the problem of searching for a girth maximum $(m,r)$ BTU as a question in Extremal Graph Theory by raising two related questions.
\begin{enumerate}
\item Given girth $g$ and  $r\in \mathbb{N}$ , what is the minimum value of  $m$ such that a $(m,r)$ BTU has girth $g$  .
\item Given  $m,r\in \mathbb{N};m\gg r$, what is the maximum attainable girth for a $(m,r)$ BTU ?
\end{enumerate}

\section{Maximum Attainable Girth}
\subsection {Maximum Attainable Girth for a $(m,r)$ BTU}
We denote the maximum Attainable Girth for a $(m,r)$ BTU as a function  $g_{\mathit{max}}:\{\mathbb{N}\cup \{0\}\}^{2}\to \mathbb{N}\cup \{0\}$.

\begin{theorem}
The maximum attainable girth of a  $(m,r)$  BTU satisfies the inequality  $g_{\mathit{max}}(m,r)<2\ast k$  where  $k\in \mathbb{N}$  is obtained by minimizing  $b\in \mathbb{N}$  such that  $m=b\ast k^{r-1}$ for $r\ge 3$ .
\end{theorem}
\begin{proof}
From the optimal partition result from $[2]$ for a $(m, r)$ BTU for $r\ge 3$, we obtain that the maximum possible lenght of the maximum known cycle is $2 *k$, where the optimal paritions are $\beta _{i}${ refers to } $\sum _{j=1}^{r-1-i}b\ast k^{i}=b\ast k^{r-1}$ {for } $1\le i\le r-1$ and $k$ is obtained by minimizing $b\in \mathbb{N}$  such that  $m=b\ast k^{r-1}$ for $r\ge 3$. We now need to show that $g_{\mathit{max}}$ cannot equal $2*k$ for $r\ge 3$. This follows because of micro-partition cycles defined in $[2]$ and their combinations which do not permit $g_{\mathit{max}}$ to equal $2*k$ for $r\ge 3$. Hence, the result follows.
\end{proof} 

\section {High Level Description Of Enumeration Based Search from $[16]$}
\subsection{Enumeration Based Search algorithm for girth maximum $(m, r)$ BTU for\ $r>3$} 
We find $b,k\in \mathbb{N}$ such that $b$ is the smallest integer satisfying $m=b\ast k^{r-1}$; \\
for(  $i=2;i<r;i$++) \{ \\
\  $p_{i}=C_{j};\mathit{min}(b\ast k^{i-1}-j,j)>b\ast k^{i-2}$ such that $(j,b\ast k^{i-1},b\ast k^{i-1}-j)$ are relatively prime; \\
if( $i$ \ \ == \  $2$ \ ) \\
$p_{i-1}=I_{b\ast k^{i-1}};$  \\
\ else \{ \\
\ Rearrange the $(b\ast k^{i-1},i)$ BTU such that \ \  $p_{i-1}=I_{b\ast k^{i-1}}$; \\
Find $q_{i-2}\in S_{b\ast k^{i-2}}$ such that it maximizes girth of  $(b\ast k^{i-1},i)$ BTU is formed by  $p_{1},\ldots ,p_{i-2}\in S_{b\ast k^{i-1}};p_{x}=k\ast q_{x};1\le x\le i-2$; \\
if(i != r -- 1) \\
Scale permutations $p_{y}=k\ast q_{y};1\le y\le i$; \\
\} \\
\}

\subsection{Enumeration Based Search algorithm for a girth maximum $(m, 3)$ BTU  where $ m = b \ast k^{2}$}
We find $b,k\in \mathbb{N}$ such that $b$ is the smallest integer satisfying $m=b\ast k^{2}$; \\
for(  $i=2;i<3;i$++) \{ \\
\  $p_{i}=C_{j};\mathit{min}(b\ast k^{i-1}-j,j)>b\ast k^{i-2}$ such that $(j,b\ast k^{i-1},b\ast k^{i-1}-j)$ are relatively prime; \\
if( $i$ \ \ == \  $2$ \ ) \\
$p_{i-1}=I_{b\ast k^{i-1}};$  \\
\ else \{ \\
\ Rearrange the $(b\ast k^{i-1},i)$ BTU such that \ \  $p_{i-1}=I_{b\ast k^{i-1}}$; \\
Find $q_{1} \in S_{b\ast k}$ such that a girth maximum  $(b\ast k^{2},3)$ BTU is formed by  $p_{1}, p_{2}, p_{3}\in S_{b\ast k^{i-1}};p_{1}=k\ast q_{1}$; \\
\}

\subsection {Reorganizing the $(b \ast k^{i -1}, i)$ BTU such that $p_{i -1} = I_{b \ast k^{i -1}}$ }
Without loss of generality,  we apply suitable permutations on depth and permutations labels on the $(b \ast k^{i -1}, i)$ BTU in order to obtain $p_{i -1} = I_{b \ast k^{i -1}}$.  Permutations on depth and permutations labels have been explained and defined in $[1]$ and preserve isomorphism since they correspond to row permutations and column permutations on the matrix representation of the $(b \ast k^{i -1}, i)$ BTU.

\section {Detailed Description Of Enumeration Based Search for a girth maximum  $(k^2, 3)$ BTU}

To find permutation a $q_{1} \in S_{k} $ such that a girth maximum  $(k^{2},3)$ BTU is formed by  $\{p_{1},..,p_{3}\}$ \{  \\
We enumerate all permutations  $q_{1}$ with node at depth $1$ fixed,
such that partition between  $q_{1}$  and  $q_{2} = I_{k} $ is $(k) \in P_2(k)$ ; \\
for(each enumerated permutation\  $q_{1}$ ) \{ \\
We scale up $q_{1}$ by  $k$ and  $p_{2}=I_{k^{2}}; \\
p_{3}=C_{j}$ where  $(j,k^{2}, k^{2}\text{--}j)$ are relatively prime; \\
We compute the girth of this $(k^2, 3)$ BTU; \\
\} \\
We choose permutation $q_{1}$ that gives us the best girth;

\section {Detailed Description Of Enumeration Based Search for a girth maximum  $(m, r)$ BTU where $b,k\in \mathbb{N}$ such that $b$ is the smallest integer satisfying $m=b\ast k^{r - 1}$ }

To find permutations $\{q_{1},{\dots},q_{i\text{--}2}\}\in S_{b\ast k^{i-2}}$ such that a girth maximum  $(b\ast k^{i-1},i)$ BTU is formed by  $\{p_{1},..,p_{i}\}$ \{  \\
We enumerate all permutations  $q_{i\text{--}2}$ with node at depth $1$ fixed,
such that partition between  $q_{i\text{--}2}$  and 
$q_{i - 1} = I_{b \ast k^{i-2}} $ is $(b\ast k^{i-2})$ ; \\
for(each enumerated permutation\  $q_{i\text{--}2}$ ) \{ \\
We permute $\{q_{1},{\dots},q_{i\text{--}3}\}$ such that all partitions between any two permutations in the set \  $\{q_{1},{\dots},q_{i-2},q_{i\text{--}1}\}$ are preserved;  \\
We scale up $\{q_{1},{\dots},q_{i\text{--}2}\}$ by  $k$ and  $p_{i\text{--}1}=I_{b\ast k^{i-1}};p_{i}=C_{j}$ where  $(j,b\ast k^{i-1},b\ast k^{i-1}\text{--}j)$ are relatively prime; \\
We compute the girth of this $(b\ast k^{i-1},i)$ BTU; \\
\} \\
We choose permutation $q_{i\text{--}2}$ that gives us the best girth;

\section {Algorithm to Find Permutations Of $\{q_{1},{\dots},q_{i\text{--}2}\}$}
We permute  $\{q_{1},{\dots},q_{i\text{--}2}\}$ such that all partitions between any two permutations in the set  $\{q_{1},{\dots},q_{i\text{--}2},,I_{b\ast k^{i-2}}\}$ are preserved \\
for( $j=2;j<b\ast k^{i-2};j$++) \{  \\
 $d$ = Label at depth  $j$ of  $q_{i-2}$; \\
Permutations On Depth $(d,j)$; \\
Permutations On Labels $(d,j)$; \\
We calculate the partition between permutations  $k\ast q_{i\text{--}2}$ and  $p_{i}$ and girth; \\
We accept the change to  $\{q_{1},{\dots},q_{i\text{--}2}\}$  if it improves the girth; \\
\} \\
 $q_{i\text{--}1}$ returns to  $I_{b\ast k^{i-2}}$ after each run of the loop.

\section {Experimental Results for Implementation Of Enumeration Based Search}
Girth obtained for various values of  $m$ and  for $r = 3$  has been shown Table $1$ . We find that the values of $m$ for a given value of girth $g$ lie between the lower bound for $m$ and improved lower bound for $m$ from $[13]$. The execution time is too long for $k >10$ due to the algorithm being in EXPTIME.

\begin{table}
\caption{Girth obtained for various of $m$ and for $r = 3$ from Implementation}
\begin{tabular}{llllll}
\hline\noalign{\smallskip}
$k$ & $m$ & $r$ & $g$ \\
\noalign{\smallskip}
\hline
\noalign{\smallskip}
5 & 25 & 3 & 8 \\
6 & 36 & 3 & 8  \\
7 & 49 & 3 & 10 \\
8 & 64 & 3 & 10 \\
9 & 81 & 3 & 10  \\
10 & 100 & 3 & 10 \\
\hline
\end{tabular}
\end{table}

\section {Bound from $[12]$}
For $q$ being a power of a prime $k\ge 3$ , Lazebnik in  $[12]$ describes explicit construction of a \  $q$ {}-regular bipartite graph on  $v=2\ast q^{k}$ vertices with girth  $g\ge k+5$.

If we consider this as a $(m,r)$ BTU, we get  $r$ a power of a prime and  $m=r^{k};k\ge 3$, girth  $g\ge \log _{r}(m)+5$. For  $g\ge 12$, we obtain  $\log _{r}(m)\ge 7$ which gives us  $m\ge r^{7}$ and we hence obtain $m\ge 3^{7}=343\ast 9=3087$.

\section {Lower bounds from \ $[9]$}
We quote the main theoem from $[9]$, "Let $G= (V_{L}, V_{R}, E)$ be a bi-partite graph of girth $g = 2 \ast r$, with $n_{L}=\left(V_{L})\right|$ and $n_{R}=\left(V_{R})\right|$, the number of vertices on the left and right sides, and $m=\left(E)\right|$ the number of edges. Assume further that all vertex degrees in $G$ are $\ge 2$ Then: $n_{L}\ge \sum _{i=0}^{r_{1}\text{--}1}(\Lambda _{R})^{\mathit{ceil}(i/2)}(\Lambda _{L})^{\mathit{floor}(i/2)}$ and  $n_{R}\ge \sum _{i=0}^{r_{1}\text{--}1}(\Lambda _{L})^{\mathit{ceil}(i/2)}(\Lambda _{R})^{\mathit{floor}(i/2)}$ where  $\Lambda _{R}=\prod _{v\in V_{R}}(d_{v}\text{--}1)^{d_{v}/m}$,  $\Lambda _{L}=\prod _{v\in V_{L}}(d_{v}\text{--}1)^{d_{v}/m}$ and $d_{v}$ is the degree of vertex $v$."

\subsection {From another form of the bound in $[9]$}
From another form of the bound in  $[9]$, \ we obtain $n_{L}\ge \sum _{i=0}^{r\text{--}1}(d_{R}\text{--}1)^{\mathit{ceil}(i/2)}(d_{L}\text{--}1)^{\mathit{floor}(i/2)}$ and  $n_{R}\ge \sum _{i=0}^{r\text{--}1}(d_{L}\text{--}1)^{\mathit{ceil}(i/2)}(d_{R}\text{--}1)^{\mathit{floor}(i/2)}$.For a $(m,r)$ BTU with girth  $g$, we obtain  $m\ge \sum _{i=0}^{g/2\text{--}1}(r\text{--}1)^{\mathit{ceil}(i/2)}(r\text{--}1)^{\mathit{floor}(i/2)}$.

Therefore,  $m\ge \sum _{i=0}^{g/2\text{--}1}(r\text{--}1)^{\mathit{ceil}(i/2)+\mathit{floor}(i/2)}$ 

For even integers  $i$,  $\mathit{ceil}(i/2)+\mathit{floor}(i/2)=i$

For odd integers  $i$ ,  $\mathit{ceil}(i/2)+\mathit{floor}(i/2)=(i+1)/2+(i-1)/2=i$ 

Therefore,  $m\ge \sum _{i=0}^{g/2\text{--}1}(r\text{--}1)^{i}=(r-1)^{g/2-1+1}=\frac{(r-1)^{g/2}-1}{(r-2)}$ 

Putting  $r=3$ and  $g=12$ we get  $m\ge \frac{(2)^{6}-1}{3-2}=63$. 
Putting  $r=3$  and  $g=10$ we get  $m\ge \frac{(2)^{5}-1}{3-2}=31$. 
Putting  $r=3$  and  $g=8$ we get  $m\ge \frac{(2)^{4}-1}{3-2}=15$ .

\subsection {From Main Theorem in $[9]$}
Derived from the main theorem, 
From  $[9]$, \  $n_{L}\ge \sum _{i=0}^{r_{1}\text{--}1}(\Lambda _{R})^{\mathit{ceil}(i/2)}(\Lambda _{L})^{\mathit{floor}(i/2)}$ and  $n_{R}\ge \sum _{i=0}^{r_{1}\text{--}1}(\Lambda _{L})^{\mathit{ceil}(i/2)}(\Lambda _{R})^{\mathit{floor}(i/2)}$ .
For a $(m,r)$  BTU with girth  $g$ , we obtain, $\Lambda _{R}=\{(r\text{--}1)^{r/(m\ast r)}\}^{m}=r-1$ and  $\Lambda _{L}=\{(r\text{--}1)^{r/(m\ast r)}\}^{m}=r-1$ .
Thus, \  $m\ge \sum _{i=0}^{g/2\text{--}1}(r\text{--}1)^{\mathit{ceil}(i/2)}(r\text{--}1)^{\mathit{floor}(i/2)}$.
Therefore,  $m\ge \sum _{i=0}^{g/2\text{--}1}(r\text{--}1)^{\mathit{ceil}(i/2)+\mathit{floor}(i/2)}$ .
For even integers  $i$ ,  $\mathit{ceil}(i/2)+\mathit{floor}(i/2)=i$ .
For odd integers  $i$ ,  $\mathit{ceil}(i/2)+\mathit{floor}(i/2)=(i+1)/2+(i-1)/2=i$ .
Therefore,  $m\ge \sum _{i=0}^{g/2\text{--}1}(r\text{--}1)^{i}=(r-1)^{g/2-1+1}=\frac{(r-1)^{g/2}-1}{(r-2)}$ 
Putting  $r=3$  and  $g=12$ we get  $m\ge \frac{(2)^{6}-1}{3-2}=63$ .
Putting  $r=3$  and  $g=10$ we get  $m\ge \frac{(2)^{5}-1}{3-2}=31$ .
Putting  $r=3$  and  $g=8$ we get  $m\ge \frac{(2)^{4}-1}{3-2}=15$ .

\section {Other Related Research}
Irregular LDPC codes with girth  $20$ in $[11]$ and Regular LDPC codes of girth at least $10$ from $[10]$ .

\section {Results from $[15]$}
We quote Theorem from $[15]$ for even values of $g$ since our current interest is only in bi-partite graphs. "For $g \ge 3$ and $\delta \ge 3$ put 
$n_{0}(g,\delta ) = \frac{2\ast \{(\delta \text{--}1)^{(g/2)}\text{--}1\}}{(\delta \text{--}2)}$ if  $g$ is even. Then a graph $G$ with minimal degree $\delta$ and girth $g$ has at least $n_{0}(g,\delta )$ vertices."
We use this result to compute $n_{0}(g,\delta )$ for $\delta = 3$ and various values of $g$ in Table $2$ by simplifying the equation as 
$n_{0}(g, 3 ) = {2\ast \{(2)^{g/2}\text{--}1\}}$

\begin{table}
\caption{Minimum value of $n_{0}(g,3)$ for different girths $g$ for $\delta =3$ from $[15]$ }
\begin{tabular}{llllll}
\hline\noalign{\smallskip}
$g$ && $n_{0}(g,3)$  \\
\noalign{\smallskip}
\hline
\noalign{\smallskip}
4 && 6   \\
6 && 14\\
8 && 30  \\
10 && 62  \\
12 && 126  \\
14 && 254 \\
\hline
\end{tabular}
\end{table}

\section {Results from $[13]$}
We quote theorems from  $[13]$ .
\begin{enumerate}
\item "Given  $\delta \ge 3$ and  $g\ge 3$ , there exists a $G^{n}$,  $n\le (2\ast \delta )^{g}$ with minimal degree of at least $\delta $ and girth of at least  $g$".
\item "Lower Bound $n(g,\delta )\ge \frac{1+\delta \{(\delta \text{--}1)^{(g\text{--}1)/2}\text{--}1\}}{(\delta \text{--}2)}$ if  $g$ is odd. 
 $n(g,\delta )\ge \frac{\{(\delta \text{--}1)^{(g)/2}\text{--}1\}}{(\delta \text{--}2)}$  if  $g$ is even. 
Equality holds for  $\delta =3$ and  $g=\{3,4,5,6,7,8\}$ and  $g=4$,  $\delta \ge 3$" .
\item "If $g$ is odd, \  $n(g+1,\delta )\le 2\ast n(g,\delta )$".
\item "Upper Bound $n(g,\delta )\le \frac{2\ast \{(\delta \text{--}1)^{(g\text{--}1)}\text{--}1\}}{(\delta \text{--}2)}$  if  $g$ is odd. 
\  $n(g,\delta )\le \frac{4\ast \{(\delta \text{--}1)^{(g\text{--}2)}\text{--}1\}}{(\delta \text{--}2)}$ if  $g$ is even". 
\item "Let  $m\ge \sum _{i=0}^{g\text{--}2}(\delta -1)^{i}=\frac{(\delta -1)^{g\text{--}1}\text{--}1}{(\delta -2)}$ be an integer. Then there exists a  $\delta $ {}-regular graph of order $2\ast m$ and girth of at least  $g$" .
\item "Most significant improvement of the bound for  $\delta =3$ ,  $n(g,3)\le 2^{g-1}$".
\end{enumerate}

\section {Bound derived from $[13]$}
We derive the following bound from $[13]$ , $\frac{(\delta -1)^{g/2}-1}{\delta -2}\le n(g,\delta )\le \frac{4\ast {(\delta -1)^{g-2}-1}}{\delta -2}$ for the minimum order $n(g, \delta)$ where $g$ is its girth and $\delta$ is its degree. By putting  $\delta =3$, we obtain a simplified form of the above equation, ${(2)^{g/2}-1}\le n(g,3 )\le 4\ast {(2)^{g-2}-1}$ which could be further simplified as \\ ${2^{g/2}-1}\le n(g,3 )\le  {2^{g}-1}$.
We calculate the bounds for $\delta =3$ and the improved upper bound corresponds to $n(g,3)\le 2^{g-1}$ from $[13]$ in Table $3$. 

\begin{table}
\caption{Lower Bound, Upper Bound and Improved Upper Bound for $n(g,3)$ for different girths $g$ for $\delta =3$ from $[13]$ }
\begin{tabular}{llllll}
\hline\noalign{\smallskip}
$g$ && Lower Bound $n(g,3)$ & Upper Bound $n(g,3)$ & Improved upper bound $n(g,3)$ \\
\noalign{\smallskip}
\hline
\noalign{\smallskip}
4 && 3  & 15 & 8 \\
6 && 7 & 63  & 32\\
8 && 15 & 255 & 128  \\
10 && 31  & 1023 & 512\\
12 && 63 & 4095 & 2048 \\
14 && 127 & 16383 & 8192  \\
\hline
\end{tabular}
\end{table}

\section {Analysis for  $[17]$ and  $[18]$}

From $[17]$, we quote the following result, 
"If the degree is  $D\ge 3$ and girth  $g=2\ast r+1;r\ge 2$, a simple lower bound for number of vertices of a regular graph is given by $n_{o}(g,D)=1+\frac{D}{D\text{--}2}((D\text{--}1)^{r}-1)$." \\ For  $D=3$ we simplify the equation as follows  $n_{o}(g,3)=1+3((2)^{(g-1)/2}-1)$. While the exponent is similar to the lower bound in  $[13]$, we cannot apply the result as the girths take odd values and do not directly apply for bi-partite graphs. 

\section {Analysis for  $[19]$}

We analyze the girths obtained for various size of the matrices from $[19]$ in Table $4$. However, these matrices have irregular degrees and hence a direct comparison with our obtained results might not be possible.

\begin{table}
\caption{Girth obtained for various size of the matrices in  $[19]$}
\begin{tabular}{llllll}
\hline\noalign{\smallskip}
Girth & Minimum $N$ \\
\noalign{\smallskip}
\hline
\noalign{\smallskip}
6 & 5\\
8 & 9\\
10 & 39\\
12 & 97\\
\hline
\end{tabular}
\end{table} 

\section {Analysis for  $[14]$ }

We quote from $[14]$, "Ramanujan graphs $X^{p,q}$ {are } $p+1$ regular Cayley graphs of the group $\mathit{PSL}(2,\mathbb{Z}/q\mathbb{Z})$ if the Legendre symbol $(\frac{p}{q})=1$ and of \  $\mathit{PGL}(2,\mathbb{Z}/q\mathbb{Z})$ {if the Legendre symbol } $(\frac{p}{q})=-1$ . $X^{p,q}$ is bi-partite of order  $n=\left|(X^{p,q})\right|=q\ast (q^{2}-1)$ and a bound on the girth is given by the equation, $g(X^{p,q})\ge 4\log _{p}(q)-\log _{p}(4)$".

Putting $p = 2$ in order to get degree $k = p + 1 = 3$, we obtain the inequality $ g \ge 4\log _{2}(q)-\log _{2}(4)$ which can be simplified as $ (g + 2)/4 \ge \log _{2}(q) $ in order to obtain $2^{(g+2)/4} \ge q $.

For each value of girth $g$, we calculate the minimum value of $q$ such that $q \ge 2^{(g+2)/4}$ and the Legendre symbol  $(p/q) = -1$ and then calculate $n = q \ast (q^2 - 1) $ for $p = 2$ and degree $k = 3$ in Table $5$.

\begin{table}
\caption{Analysis for  $[14]$}
\begin{tabular}{llllll}
\hline\noalign{\smallskip}
Girth & min $q$, $q \ge 2^{(g+2)/4}$, $(p/q) = -1$  & $n = q \ast (q^2 - 1) $ & Chosen $p$ & Degree  $k=p+1$ \\
\noalign{\smallskip}
\hline
\noalign{\smallskip}
6 & 5 & 120 & 2 & 3 \\
8 & 11 & 1320 & 2 & 3  \\
10 & 11 & 1320 & 2 & 3 \\
12 & 13 & 2184 & 2 & 3 \\
\hline
\end{tabular}
\end{table}

\section{Conclusion}
Our implementation for the Enumeration based Search for a girth maximum $(m, r)$ BTU finds the maximum attainable girth of a  $(m,r)$ BTU for  $r = 3$ and various values of $m$.  The values of $m$ for a given girth $g$ are within the known mathematical bounds for regular bi-partitite graphs from the available literature. When we compare our results with bounds for more general graphs, or graphs with irregular graphs, a direct comparison may not possible since it is well known that for a given $g$ and average degree, a lower number of vertices can be reached for irregular graphs.  

\begin{thebibliography}{[MT1]}
%
\bibitem[1]{}
Vivek S Nittoor, Reiji Suda,:
Balanced Tanner Units And Their Properties,
arXiv:1212.6882 [cs.DM].
%
\bibitem[2]{}
Vivek S Nittoor, Reiji Suda,:
Partition Parameters for Girth Maximum $(m ,r)$ BTUs,
arXiv:1212.6883 [cs.DM].
%
\bibitem[3]{}
Vivek S Nittoor, Reiji Suda,:
Parallelizing A Coarse Grain Code Search Problem Based upon LDPC Codes on a Supercomputer,
Proceedings of 6th International Symposium on Parallel Computing in Electrical Engineering (PARELEC 2011), Luton, UK, April 2011.
%
\bibitem[4]{}
R. M. Tanner,:
A recursive approach to low complexity codes,
IEEE Trans on Information Theory, vol. IT-27, no.5, pp. 533-547, Sep 1981. 
%
\bibitem[5]{}
C.E. Shannon,:
A Mathematical Theory of Communication,
Bell System Technical Journal, vol. 27, pp 379-423, 623-656, July, October, 1948.
%
\bibitem[6]{}
D. J. C. MacKay, R. M. Neal,:
Near Shannon limit performance of low density parity check codes,
 Electron. Lett., vol. 32, pp. 1645--1646, Aug. 1996.
%
%
\bibitem[7]{}
William E. Ryan, Shu Lin,:
Channel Codes Classical and Modern,
Cambridge University Press, 2009.
%
%
\bibitem[8]{}
F. Harary,:
Graph Theory,
Addison-Wesley, 1969.
%
%
\bibitem[9]{}
Shlomo Hoory,:
The Size Of Bipartite Graphs with a Given Girth,
Journal Of Combinatorial Theory, Series B 86, 215-220 (2002).
%
%
\bibitem[10]{}
Fan Zhang, at al:
High Girth LDPC Code Construction Based on Combinatorial Design,
Vehicular Technology, IEEE Conference - VTC -Spring , vol. 1, pp. 591-594 Vol. 1, 2005.
%
%
\bibitem[11]{}
J.M.F. Moura, et al:
Structured Low-Density Parity-Check Codes
IEEE Signal Processing Magazine \ vol. 21, no. 1, pp. 42-55, 2004.
%
%
\bibitem[12]{}
F. Lazebnik, V.A. Ustimenko,:
Explicit Construction of Graphs with arbitrary large Girth and of Large Size,
Discrete Applied Mathematics 60 (1995), 275-284.
%
\bibitem[13]{}
Bollobas,:
Extremal Graph Theory
Academic Press, London, 1978.
%
\bibitem[14]{}
A. Lubotzky, R. Philips, P. Sarnak,:
Ramanujan Graphs,
Combinatorica 8(3) (1988) 261-277.
%
%
\bibitem[15]{}
Bollobas,:
Modern Graph Theory
Springer, London, 1998.
%
\bibitem[16]{}
Vivek S Nittoor, Reiji Suda,:
Enumeration Based Search Algorithm For Finding A Regular Bi-partite Graph Of Maximum Attainable Girth For Specified Degree And Number Of Vertices,
Available at cs.DM arxiv.
%
\bibitem[17]{}
Mirka Miller, Jozef Siran,:
Moore graphs and beyond: A survey of the degree/diameter problem,
Electronic Journal Of Combinatorics, Dynamic Survey D, Vol. 14, 2005.
%
\bibitem[18]{}
N.I. Biggs,:
Girth, valency and excess,
 Linear Algebra Appl. 31 (1980) 55–59.
%
\bibitem[19]{}
Michael E. O’Sullivan,:
Algebraic Construction of Sparse Matrices With Large Girth,
IEEE Transactions on Information Theory, Vol. 52, No. 2, Feb 2006.
%
\bibitem[19]{}
S. Ramanujan,:
On certain arithmetical functions,
Trans. Camb. Phil. Soc. 22(1916), 159-184.
%

\end{thebibliography}
\end{document}